\documentclass[a4paper,UKenglish,cleveref, autoref, thm-restate]{lipics-v2021}

\title{The Complexity of Checking Non-Emptiness in
Symbolic Tree Automata} 

\author{Rodrigo Raya}{School of Computer and Communication Science \\ École Polytechnique Fédérale de Lausanne (EPFL), Switzerland \\
rodrigo.raya@epfl.ch}{}{}{}

\authorrunning{Rodrigo Raya} 

\Copyright{Rodrigo Raya} 

\ccsdesc[500]{Theory of computation~Automated reasoning}

\keywords{automata, computational complexity, model theory} 

\category{} 

\relatedversion{} 

\usepackage{forest}
\newcommand\mpn[3]{\nodepart{one}   #1
                   \nodepart{two}   #2
                   \nodepart{three} #3}

\usepackage{graphicx}

\usepackage[inline]{enumitem}

\usepackage{ amsmath }

\usepackage{ amsfonts }

\usepackage{ stmaryrd }

\usepackage{graphicx}
\usepackage{float}
\usepackage{tikz}
\usetikzlibrary{positioning,automata} 

\tikzset{%
  zeroarrow/.style = {-stealth,dashed},
  onearrow/.style = {-stealth,solid},
  c/.style = {circle,draw,solid,minimum width=2em,
        minimum height=2em},
  r/.style = {rectangle,draw,solid,minimum width=2em,
        minimum height=2em}
}

\usepackage{booktabs}

\usepackage{braket}

\usepackage[strings]{underscore}

\usepackage{dirtytalk}

\usepackage{ifthen}

\usepackage{float}

\usepackage{cite}

\usepackage{tikz}
\usepackage{tikz-qtree}
\usetikzlibrary{chains,shapes.multipart}
\usetikzlibrary{automata, arrows.meta, positioning, quotes}

\usepackage[strings]{underscore}

\newcommand{\cc}{\mbox{C}}

\newcommand{\np}{\mbox{NP}}

\newcommand{\N}{\mathbb{N}}

\newcommand{\qfbapa}{\textsf{QFBAPA}}

\newcommand{\parikh}{\textsf{Parikh}}

\begin{document}

\maketitle

\nolinenumbers

\begin{abstract}
We study the satisfiability problem of symbolic tree automata and decompose it into the satisfiability problem of the existential first-order theory of the input characters and the existential monadic second-order theory of the indices of the accepted words. We use our decomposition to obtain tight computational complexity bounds on the decision problem for this automata class and an extension that considers linear arithmetic constraints on the underlying effective Boolean algebra.
\end{abstract}

\section{Introduction}
\label{section:intro}
The purpose of this paper is to expose certain analogy that can be made between the computational complexity analysis of the decision problem of symbolic tree automata and the decision problem of a class of logical structures known as power structures in model theory. 

The notion of symbolic automata appeared first in \cite{watson_implementing_1996}. But it was not until \cite{veanes_rex_2010} that this class of automata regained attention. Symbolic automata have been used in a variety of applications including the analysis of regular expressions \cite{veanes_rex_2010, dantoni_minimization_2014}, string encoders \cite{hooimeijer_fast_2011, dantoni_extended_2015, hu_automatic_2017}, functional programs \cite{dantoni_fast_2015}, code generation, parallelisation \cite{saarikivi_fusing_2017} and symbolic matching \cite{saarikivi_symbolic_2019}. The specific subclass of symbolic tree automata (STAs) was later studied in the sequence of publications \cite{kaminski_tree_2008, fulop_forward_2014, veanes_symbolic_2015, dantoni_fast_2015, dantoni_minimization_2016}.

Several theoretical investigations have been carried out on computational aspects of the symbolic automata model, including \cite{dantoni_minimization_2014, tamm_theoretical_2018, argyros_learnability_2018}. In particular, the authors of \cite{veanes_monadic_2017} observed that such an automata model had been studied previously by Bès in  \cite{bes_application_2008}. In his paper, Bès introduced a class of multi-tape synchronous finite automata whose transitions are labelled by first-order formulas. He then proved various properties of the languages accepted by such automata including closure under Boolean, rational, and the projection operations, logical characterisations in terms of MSO logic and the Eilenberg-Elgot-Shepherdson formalism \cite{eilenberg_sets_1969} as well as decidability properties. Remarkably, the paper showed that the notion of recognisability for such automata coincides with that of definability for certain generalised weak powers, first-studied by Feferman and Vaught \cite{feferman_first_1959}. In the concluding remarks of his work, Bès noted that \say{all results in this paper can be extended to the case of infinite words as well as (in)finite binary trees, by relying on classical decidability results for MSO theories}.

The techniques of Feferman and Vaught allow decomposing the decision problem for the first-order theory of a product of structures\footnote{The notion of the theory of a structure for first-order and monadic second-order theories is standard in model theory. We refer the interested reader to the book of Hodges \cite{hodges_model_1993} for the details.}
, $Th(\prod_i \mathcal{M}_i)$ into the first-order theory of the structures $\mathcal{M}_i$, $Th(\mathcal{M}_i)$, and the monadic second-order theory of the index set $I$, $Th^{mon}(\langle I, \ldots \rangle)$, where the structure $\langle I, \ldots \rangle$ may contain further relations such as a finiteness predicate, a cardinality operator, etc. If the theory of the components $Th(\mathcal{M}_i)$ is decidable for each $i \in I$, then the decision problem reduces to that of the theory $Th^{mon}(\langle I, \ldots \rangle)$. To analyse these structures, Feferman and Vaught extended results going back to Löwenheim, Skolem and Behmann \cite{lowenheim_uber_1915, skolem_untersuchungen_1919, behmann_beitrage_1922}. Technically, the decomposition is expressed in terms of so-called reduction sequences.

It is known \cite{dawar_model_2007} that many model-theoretic constructions incur in non-elementary blow-ups in the formula size. This includes the case of the size of the Feferman-Vaught reduction sequences in the case of disjoint unions. Perhaps for this reason, no computational complexity results have been obtained for the theory of symbolic tree automata and related models. Instead, the results in the literature \cite{dantoni_power_2017, fisman_complexity_2021, dantoni_automata_2021} refer to the decidability of the satisfiability problem of the monadic predicates or provide asymptotic run-times rather than a refined computational complexity classification, which could help to evaluate the speed and scale to which we can hope to solve the satisfiability problem.

As a \textbf{main contribution}, we show how to reduce the satisfiability problem for symbolic tree automata to the satisfiability problem of the existential first-order theory of the input characters and the existential monadic second-order theory of the indices. This decomposition allows us to derive tight complexity bounds for the decision problem of the automaton in the precise sense of Corollary~\ref{cor:complexity}. We then study an extension of the formalism of symbolic tree automata which also imposes linear arithmetic constraints on the cardinalities of the Venn regions of the underlying effective Boolean algebra. In particular, this extension allows expressing the number of occurrences of a particular kind of letter in a word. We show in Corollary~\ref{cor:complexitycard} that the computational complexity of the corresponding satisfiability problem is the same as the one for the simpler model without cardinalities. Similar extensions for related models of automata are considered in the literature on data words \cite{figueira_reasoning_2022}.

\textbf{Organisation of the paper.} Section~\ref{section:symbolic} introduces symbolic tree automata. Section~\ref{section:feferman} gives a preliminary Feferman-Vaught decomposition of symbolic tree automata in terms of the theory of the elements and the theory of the indices. Section~\ref{section:decision} describes the decision procedure with which, in Section~\ref{section:with}, after presenting the quantifier-free theory of Boolean algebra with Presburger arithmetic, we obtain the tight complexity bounds announced. Section~\ref{section:cardinalities} describes the extension of symbolic tree automata that uses linear arithmetic constraints over the cardinalities of the automaton's underlying effective Boolean algebra and proves the corresponding upper bounds for the associated satisfiability problem. Section~\ref{section:conclusion} concludes the paper.

\section{Symbolic Tree Automata (STA)}
\label{section:symbolic}
In this section, we introduce the automata model that we will study in the rest of the paper. 

Unlike traditional tree automata, symbolic automata read input characters over a not necessarily finite domain, which has the structure of a Boolean algebra of sets defined by a family of monadic predicates. To ensure compatibility with the Boolean algebra operations, the family of monadic predicates defining the sets needs to be closed under propositional operations and contain formulae denoting the empty set and the universe. Furthermore, the definition requires that checking satisfiability of the monadic predicates is decidable. In later sections, we will refine this assumption with different complexity-theoretic bounds.

\begin{definition}[\cite{dantoni_automata_2021}]
\label{def:eba}
An effective Boolean algebra $\mathcal{A}$ is a tuple 
\[
\left(\mathfrak{D}, \Psi, \llbracket \cdot \rrbracket, \perp, \top, \vee, \wedge, \neg\right)
\]
where $\mathfrak{D}$ is a set of domain elements, $\Psi$ is a set of unary predicates over $\mathcal{D}$ that are closed under the Boolean connectives, with $\perp, \top \in \Psi$ and $\llbracket \cdot \rrbracket: \Psi \rightarrow 2^{\mathcal{D}}$ is a function such that \begin{enumerate*}
\item $\llbracket \perp \rrbracket=\emptyset$, 
\item $\llbracket \top \rrbracket=\mathfrak{D}$, and 
\item for all $\psi, \psi_1, \psi_2 \in \Psi$, we have that
\begin{enumerate*}
\item $\llbracket \psi_1 \vee \psi_2 \rrbracket=\llbracket \psi_1 \rrbracket \cup \llbracket \psi_2 \rrbracket$
\item $\llbracket \psi_1 \wedge \psi_2 \rrbracket=\llbracket \psi_1 \rrbracket \cap \llbracket \psi_2 \rrbracket$
\item $\llbracket \lnot \psi \rrbracket=\mathfrak{D} \backslash \llbracket \psi \rrbracket$. 
\end{enumerate*}
\item Checking $\llbracket \psi \rrbracket \neq \emptyset$ is decidable. 
\end{enumerate*} 
A predicate $\psi \in \Psi$ is atomic if it is not a Boolean combination of predicates in $\Psi$. 
\end{definition}

We will give now two examples of the notion of effective Boolean algebra. From the various examples in the literature, we choose one that matches one of our initial motivations: to generalise the complexity results obtained for array theories 
in \cite{alberti_cardinality_2017, raya_vmcai_2022}. We observe that the notion of SMT algebra in \cite[Example~2.3]{dantoni_automata_2021} precisely corresponds to the language introduced in \cite[Definition~5]{raya_vmcai_2022} but omitting cardinality constraints. We take this as a first example of effective Boolean algebra. 

\begin{example} \label{example:smt}
The SMT algebra for a type $\tau$ is the tuple $(\mathcal{D}, \Psi, \llbracket \cdot \rrbracket, \bot, \top, \lor, \land, \lnot)$ where $\mathcal{D}$ is the domain of $\tau$, $\Psi$ is the set of all quantifier-free formulas with one fixed free variable of type $\tau$, $\llbracket \cdot \rrbracket$ maps each monadic predicate to the set of its satisfying assignments, $\bot$ denotes the empty set (which can be represented by the formula $x \neq x$), $\top$ denotes the universe $\mathcal{D}$ (which can be represented by the formula $x = x$) and $\lor, \land, \lnot$ denote the Boolean algebra operations of union, intersection, and complement respectively (which can be represented by the propositional operations on quantifier-free formulae). 
\end{example}

In applications, it is often useful to consider effective Boolean algebras whose generating monadic predicates use particular representations of formulae. In particular, Example~\ref{example:smt} can be contrasted with other representations of the monadic predicates, which consider implementation details. An example of the latter is the BDD effective Boolean algebra described in \cite{dantoni_minimization_2016} which assumes that the set of elements of the underlying domain are expressed using binary decision diagrams \cite{bryant_graph-based_1986}.

\begin{example} \label{example:bit}
The BDD algebra $\mathcal{B} = (\mathbb{N}, \Psi, \llbracket \cdot \rrbracket, \bot, \top, |, \&, \overline{\cdot})$ has the set of natural numbers $\mathbb{N}$ as its universe and $\Psi$ is the Boolean closure of BDDs $\beta_i$ such that $\llbracket \beta_i \rrbracket$ is the set of natural numbers such that the i-th bit of $n$ in binary representation is one, $\bot$ denotes the BDD representing the empty set, $\top$ denotes the BDD representing the universal set and $|, \&, \overline{\cdot}$ denote the Boolean algebra operation of union, intersection, and complement. For instance, $\beta_3 \land \overline{\beta_0}$ denotes the set of numbers matching the binary bit-pattern $\ldots 1 \cdot \cdot 0$, which is satisfied by $8, 24, \ldots$.
\end{example}


We now introduce the automata model we will investigate in the paper. As in \cite{veanes_symbolic_2015}, we will assume that our automata read binary trees. 

\begin{definition}
A binary $\Sigma$-tree is a function $\tau: A \to \Sigma$ where $A$ is a finite subset of $\{0,1\}^*$ closed under the initial segment relation (i.e. if $uv \in A$ then $u \in A$). $\Sigma^{\#}$ is the class of all binary $\Sigma$-trees. $\Lambda$ the function with domain $\emptyset$ also known as the empty tree. For $\sigma \in \Sigma$ and $\tau,\tau' \in \Sigma^{\#}$, $\sigma[\tau,\tau']$ is the $\Sigma$-tree with root $\sigma$, left subtree $\tau$ and right subtree $\tau'$.
\end{definition}

The crucial difference with traditional tree automata comes in the definition of the transition relation which occurs at two levels: the symbolic level, at which we only consider the particular formula that is satisfied, and the concrete level in which we also consider the input character from the effective Boolean algebra that satisfies the predicate.

\begin{definition}[\cite{veanes_symbolic_2015}]
\label{def:treeautomata}
A symbolic tree automaton (STA) is a tuple 
\[
M = (\mathcal{A}, Q, q_0, F, \Delta)
\]
where \begin{enumerate*}
\item $\mathcal{A}$ is an effective Boolean algebra. 
\item $Q$ is a finite set of states.
\item $q_0 \in Q$ is the initial state.
\item $F \subseteq Q$ is the set of final states.
\item $\Delta \subseteq Q \times \Psi_{\mathcal{A}} \times Q \times Q$ is a finite set of transitions. 
\end{enumerate*}

A symbolic transition $\rho=\left(q_{1}, \psi, q_{2},q_3\right) \in \Delta$, also denoted $(q_{1},q_{2}) \stackrel{\psi}{\rightarrow} q_{3}$, has source states $q_{1}$ and $q_{2}$, target state $q_{2}$, and guard $\psi$. For $d \in \mathfrak{D}$, the concrete transition $(q_{1},q_{2}) \stackrel{d}{\rightarrow} q_{3}$ denotes that there exists a symbolic transition $(q_{1}, q_{2}) \stackrel{\psi}{\rightarrow} q_{3} \in \Delta$ such that $d \in \llbracket \psi \rrbracket$.

The language of $M$ at state $q \in Q$, denoted by $\mathcal{L}_q(M)$, is the smallest subset of $\mathcal{D}^{\#}$ such that \begin{enumerate*}
\item if $q \in F$ then $\Lambda \in \mathcal{L}_q(M)$, 
\item if $(q_1, \psi, q_2, q_3) \in \Delta$, $d \in \llbracket \psi \rrbracket$ and for $i \in \{1, 2\}$, $\tau_i \in \mathcal{L}_{q_i}(M)$, then $d[\tau_1, \tau_2] \in \mathcal{L}_{q_3}(M)$. 
\end{enumerate*}
The language of $M$ is $\mathcal{L}(M) = \mathcal{L}_{q_0}(M)$.
\end{definition}

We next give examples of automata running over the algebras of Example~\ref{example:smt} and \ref{example:bit}. 

\begin{example}[\cite{veanes_symbolic_2015}] 
\label{ex:smt-automaton}
We consider the language of linear arithmetic over the integers. We set three formulae $\psi_{> 0}(x) \equiv x > 0, \psi_{< 0}(x) \equiv x < 0, \psi_{= 0}(x) \equiv x = 0$ satisfied by all positive integers, all negative integers and zero, respectively. The symbolic tree automaton  with states $q_{root}, q_{-}, q_0, q_{+}, q_{\epsilon}$, final states $q_{\epsilon}$, initial state $q_{root}$ and transitions
\begin{align*}
(q_{-}, q_{+}, \psi_{= 0}, q_{root}) & & (q_{+}, q_{-}, \psi_{= 0}, q_{0}) & & (q_{\epsilon}, q_{\epsilon}, \psi_{= 0}, q_{0})\\ 
(q_{-}, q_{0}, \psi_{< 0}, q_{-}) & & (q_{\epsilon}, q_{\epsilon}, \psi_{< 0}, q_{-}) & & \\
(q_{0}, q_{+}, \psi_{> 0}, q_{+}) & & 
(q_{\epsilon}, q_{\epsilon}, \psi_{> 0}, q_{+}) & & 
\end{align*}
accepts all trees such that the root has a label $0$, its left son is a $-$ node and its right son is a $+$ node, every $-$ node has a negative label and is either a leaf or its left son is a $-$ node and its right son is a $0$ node. Similarly, every $+$ node has a positive label and is either a leaf or its right son is a $+$ node and its left son is a $0$ node. For example, the following tree would be accepted:

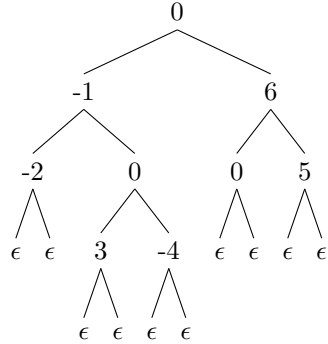
\centering
\begin{tikzpicture}
\Tree [.0 
[.-1 [.-2 [.$\epsilon$ ] [.$\epsilon$ ] ] 
[.0 [.3 [.$\epsilon$ ] [.$\epsilon$ ] ] [.-4 [.$\epsilon$ ] [.$\epsilon$ ] ] ] ]
[.6 [.0 [.$\epsilon$ ] [.$\epsilon$ ] ] [.5 [.$\epsilon$ ] [.$\epsilon$  ] ] 
] 
]
\end{tikzpicture}
\end{example}

\newsavebox{\mydiagram}
\sbox{\mydiagram}{
    \begin{tikzpicture}[node distance=.3cm and .3cm]\footnotesize
    \node (b6) {$b_6$};
    \node (b5) [below right=of b6] {$b_5$};
    \node (b4) [below right=of b5] {$b_4$};
    \node (b3) [below right=of b4] {$b_3$};
    \node (b2) [below left=of b3] {$b_2$};
    \node (b1) [below=of b2] {$b_1$};
    \node (top) [below right=of b1] {$\top$};
    \node (bot) [below left=of b1] {$\bot$};

    \draw[zeroarrow] (b6) -- (b5);
    \draw[zeroarrow] (b6) -- (b5);
    \draw[onearrow] (b5) -- (b4);
    \draw[onearrow] (b4) -- (b3);
    \draw[onearrow] (b3) -- (b2);
    \draw[zeroarrow] (b2) -- (b1);
    \draw[zeroarrow] (b1) -- (top);
    \draw[zeroarrow] (b3) -- (top);
    \draw[onearrow] (b2) -- (bot);
    \draw[onearrow] (b6) -- (bot);
    \draw[zeroarrow] (b5) -- (bot);
    \draw[zeroarrow] (b4) -- (bot);

    \end{tikzpicture}
}

\begin{example}[\cite{dantoni_minimization_2016}]
\label{ex:bdd-automaton}
We consider the language of the BDD algebra in Example~\ref{example:bit}. The following symbolic tree automaton accepts all trees whose labels represent integers such that whenever the i-th bit of such integer in binary representation is one, then the j-th bit of the integer in binary representation is also one. The automaton has a single state $q$ and a single transition rule $(q,q,\overline{\beta_i} | \beta_j, q)$:

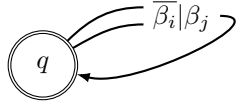
\centering
\begin{tikzpicture}[
  arrows={[scale=.75]},
  every edge quotes/.append style={inner sep=+.15em, node font=\footnotesize}]
\node (q1) [state, accepting] {$q$};
\node (bb)[above right=0mm and 1cm of q1] {$\overline{\beta_i} \vert \beta_j$};
\path[thick, in=180, at end]
  (q1) edge[out=45, "" above left] ([yshift= .75ex]bb.west)
       edge[out=30, "" below left] ([yshift=-.75ex]bb.west)
  (bb) edge[-Latex, out=0, in=-15] (q1);
\end{tikzpicture}
\end{example}

\section{Decomposition through Shared Set Variables}
\label{section:feferman}
In this section, we start with a symbolic tree automaton $M = \left(\mathcal{A}, Q, q_{0}, F, \Delta\right)$ and denote by $\psi_1, \ldots, \psi_k, \ldots$ the atomic predicates in the underlying effective Boolean algebra $\mathcal{A}$. Our first observation is that the definition of symbolic tree automaton allows assuming that the set of these predicates is finite. 

\begin{lemma} 
\label{lem:fin}
There exists a symbolic tree automaton $M' = \left(\mathcal{A}', Q, q_{0}, F, \Delta\right)$ such that $\mathcal{L}(M) = \mathcal{L}(M')$ and the cardinality of $\Psi_{\mathcal{A}'}$ is finite. 
\end{lemma}
\begin{proof}
By definition, the automaton $M$ has a finite number of transitions. We take $\Psi_{\mathcal{A}'}$ to be the Boolean closure of the predicates occurring in these transitions. It follows that $\Psi_{\mathcal{A}'}$ is a finite set. We define the remaining components of $\mathcal{A}'$ as those in the definition of $\mathcal{A}$.  Since the structure of the automaton is unchanged under this transformation, it follows that the two languages are equal, i.e. $\mathcal{L}(M) = \mathcal{L}(M')$. 
\end{proof}

From Lemma~\ref{lem:fin}, it follows that without loss of generality we may assume that $\Psi_{\mathcal{A}}$ is finite. Thus, the set of atomic predicates (see Definition~\ref{def:eba}) is finite too. In the remaining of the paper, we will work under this assumption, and we will write $\phi_1,\ldots,\phi_k$ for the generators of the effective Boolean algebra used by the symbolic finite automaton $M$. Similarly, we will write $\psi_1,\ldots,\psi_m$ for the actual predicates used in the transitions of $M$.  We will decompose the study of $\mathcal{L}(M)$ into the study of the properties of the input characters in $\mathcal{D}$ and the indexing properties induced by the transition structure of the automaton. Both kinds of properties will refer to variables representing sets of indices, to stay synchronised with each other. This methodology for combining theories had been previously studied in \cite{wies_combining_2009}.

To specify the properties of the input characters in $\mathcal{D}$, we use set interpretations of the form: 
\begin{equation} \label{eq:1}
\bigwedge_{i = 1}^k S_i = \Set{ n \in \{0,1\}^* | \phi_i(d(n)) } = \llbracket \phi_i \rrbracket
\end{equation}
where $d(n)$ is the element occurring at position $n$ in the tree $d \in \mathcal{D}^{\#}$. These sets can be pictured via a Venn diagram of interpreted sets, such as the one in Figure~\ref{fig:venn}. Each formula in $\Psi_{\mathcal{A}}$ corresponds to a particular Venn region in this diagram and can be referred to using a Boolean algebra expression on the variables $S_1, \ldots, S_k$, thanks to the set interpretation (\ref{eq:1}).

A concrete transition $(q_1,q_2) \stackrel{d}{\to} q_3$ requires a value $d \in \mathcal{D}$. This value will lie in some elementary Venn region of the diagram in Figure~\ref{fig:venn}, i.e. in a set of the form $S_1^{\beta_1} \cap \ldots \cap S_k^{\beta_k}$ where $\beta = (\beta_1,\ldots,\beta_k) \in \{0,1\}^k, S^0 := S^c$ and $S^1 := S$. We will denote such Venn region with the bit-string $\beta$. To specify the transition structure of the automaton, what is relevant to us is the region of the Venn diagram, not the specific value that it takes there. Thus, we can relabel the transitions of the automaton by the propositional formulae corresponding to the monadic predicates they held originally. 

\begin{example} \label{ex:prop}
In Example~\ref{ex:smt-automaton} the labelling predicates $\psi_{\text{=}}, \psi_{< 0}$ and $\psi_{> 0}$ would be replaced by propositional formulae $S_1, S_2$ and $S_3$. Similarly, if in Example~\ref{ex:bdd-automaton}, we take as atomic formulae the predicates $\beta_i$ then the formula $\overline{\beta_i} | \beta_j$ corresponds to the propositional formula $\lnot S_i \lor S_j$. 
\end{example}

It follows that a run of the automaton can be encoded as a tree of bit-strings $\tau: A \subseteq \{0,1\}^* \to \{0,1\}^k$ (with $A$ prefix-closed) and that these bit-strings only need to satisfy the propositional formulae corresponding to the predicates labelling the transitions of the automaton. Figure~\ref{fig:run} represents one such run over an uninterpreted Venn diagram.

We denote by $L_1,\ldots,L_m$ such propositional formulae and by $M(L_1,\ldots,L_m)$ the set of bit-string trees accepted by $M$, which we call \textit{tree tables} following the terminology of Kleene \cite{kleene_representation_1956}. Lemma~\ref{lem:automatonlanguage} observes that the language $\mathcal{L}(M)$ can be expressed in terms of set interpretations of the form~(\ref{eq:1}) and the condition: 
\begin{equation} \label{eq:2}
\exists \tau \in M(L_1,\ldots,L_m). \bigwedge_{i = 1}^k S_i = \Set{ n \in \{0,1\}^* | \tau_i(n) }
\end{equation}
where $\tau_i(n)$ denotes the $i$-th bit of $\tau$ at position $n \in \{0,1\}^*$.

This is a Feferman-Vaught decomposition in the sense that we explain next. Observe first that following the automata-logic connection discovered by Büchi \cite{buchi_weak_1960} and extended by Doner \cite{doner_tree_1970}, the tuple of sets $(S_1,\ldots,S_k)$ definable in (\ref{eq:2}) are precisely those tuples of sets definable in weak second order logic of two successors. Thus, what changes in the expression of formula~(\ref{eq:2}) is the particular representation of these relations. Lemma~\ref{lem:automatonlanguage} decomposes the satisfiability problem of symbolic tree automata into the satisfiability problem of the existential first-order theory of the input characters and the satisfiability problem of a certain representation of the monadic second-order theory of two successors. 

\begin{lemma}[Feferman-Vaught decomposition for STAs] \label{lem:automatonlanguage}
\begin{align*}
\mathcal{L}(M) = \Big\{ d \in \mathcal{D}^* \Big| &\exists \tau \in M(L_1,\ldots,L_m).\\ &\bigwedge_{i = 1}^k S_i = \Set{ n \in \{0,1\}^* | \phi_i(d(n)) } = \Set{ n \in \{0,1\}^* | \tau_i(n) } \Big\}
\end{align*}
\end{lemma}
\begin{proof}
The proof uses the definition of $\mathcal{L}(M)$ and $M(L_1,\ldots,L_m)$. For the left to right inclusion, one defines $\tau$ using the membership of the values $d(i)$ in the elementary Venn regions $\beta_i$. For the right to left inclusion, the definition of $M(L_1,\ldots,L_m)$ ensures that there is an accepting run of $M$ corresponding to the value $d$ thanks to the interpretations of the sets $S_i$. 
\end{proof}

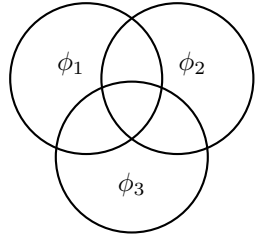
\begin{figure} 
\centering
\begin{tikzpicture}[thick]
\draw (0,0) circle (1) node[above,shift={(0,1)}] {};
\draw (1.2,0) circle (1) node[above,shift={(0,1)}] {};
\draw (.6,-1.04) circle (1) node[shift={(1.1,-.6)}] {};

\node at (-.2,.2) {$\phi_1$};
\node at (1.4,.2) {$\phi_2$};
\node at (0.6,-1.4) {$\phi_3$};
\end{tikzpicture}
\caption{A Venn diagram representing a finitely generated effective Boolean algebra with atomic predicates $\phi_1,\phi_2$ and $\phi_3$.}
\label{fig:venn}
\end{figure}

\begin{figure} 
\centering
\begin{tikzpicture}[thick]
\draw (0,0) circle (1) node[above,shift={(0,1)}] {};
\draw (1.2,0) circle (1) node[above,shift={(0,1)}] {};
\draw (.6,-1.04) circle (1) node[shift={(1.1,-.6)}] {};

\node at (-1.3,.2) {$S_1$};
\node at (2.5,.2) {$S_2$};
\node at (.6,-2.3) {$S_3$};

\node at (1.4,.2) {};
\node at (-.2,.2) {};
\node at (0.6,-1.4) {};

\draw [-to] (-.5,0) -- (.6,.4);
\draw [-to] (0,-.6) -- (.6,.34);
\draw [-to] (.6,.4) -- (1.5,.4);
\draw [-to] (.6,-.4) -- (1.5,.34);
\draw [-to] (1.4,-.7) -- (.6,-.45);
\draw [-to] (0,-1.6) -- (.55,-.45);
\end{tikzpicture}
\tikzset{every label/.style = {font=\footnotesize\sffamily\bfseries}}
    \begin{forest}
for tree = {
    rectangle split,    
    rectangle split parts=3,
    draw,
    rectangle split draw splits=false,
    inner ysep=2pt,
    parent anchor=south,
    child anchor=north,
    edge = {semithick},
    l sep=7mm,
    s sep=5mm,
            }
[\mpn{0}{1}{0}, fill=yellow!30
    [\mpn{1}{1}{0}, fill=cyan!30
            [\mpn{1}{0}{0}, fill=red!30]
            [\mpn{1}{0}{1}, fill=olive!30]
    ]
    [\mpn{1}{1}{1}, fill=green!30
            [\mpn{0}{0}{1}, fill=teal!30]
            [\mpn{0}{1}{1}, fill=orange!30]
    ]
]
    \end{forest}
\caption{A tree table accepted by a symbolic tree automaton represented over an uninterpreted Venn diagram (left) and as a bit-string tree (right). According to Doner's interpretation, the sets $A,B,C$ are $A = \{ 0,1,00,01 \}, B = \{ 
\epsilon,0,1,11 \}$ and $C = \{ 1, 01, 10, 11 \}$.}
\label{fig:run}
\end{figure}
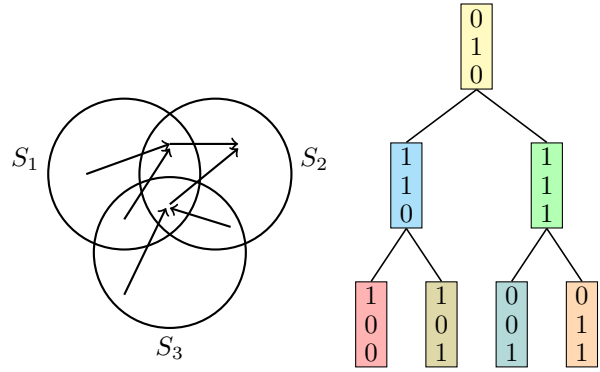
It is important to note that both sides of the equality in Lemma~\ref{lem:automatonlanguage} use essentially the same number of bits in their description, since the set $M(L_1,\ldots,L_m)$ can be described by the automaton with propositional labels or, if preferred, an equivalent regular expression. Thus, the complexity of the non-emptiness problem for both sets is the same.


In the next sections, we make use of this decomposition to devise a decision procedure for symbolic tree automata, which, will refine the existing computational complexity results for the corresponding satisfiability problem.

\section{Decision Procedure for Non-Emptiness}
\label{section:decision}
\begin{definition}
The non-emptiness problem for a symbolic tree automaton $M$ is the problem of determining whether $\mathcal{L}(M) \neq \emptyset$.
\end{definition}

By Lemma~\ref{lem:automatonlanguage}, checking non-emptiness of the language of a symbolic finite automaton reduces to checking whether the following formula is true:
\begin{equation} \label{eq:table-power}
\begin{split}
\exists &S_1, \ldots, S_k. \exists d. \bigwedge_{i = 1}^k S_i = \Set{ n \in \{0,1\}^* | \phi_i(d(n)) } \land \\ &\exists \tau \in M(L_1,\ldots,L_m). \bigwedge_{i = 1}^k S_i = \Set{ n \in \{0,1\}^* | \tau_i(n) }
\end{split}
\end{equation}

To establish the complexity of deciding formulae of the form~(\ref{eq:table-power}), we will have to analyse further the set $M(L_1,\ldots,L_m)$. Each tree table $\tau$ in $M(L_1,\ldots,L_m)$ corresponds to a \textit{symbolic table} $s$ whose entries are the propositional formulae that the bit-strings of $\tau$ satisfy. More generally, these symbolic tables are generated by the symbolic automaton obtained by replacing the predicates of the symbolic automaton by propositional formulae. The set of symbolic tables accepted by the automaton $M$ is a regular tree set and will be denoted by $M_S(L_1,\ldots,L_m)$. 

\begin{example}
The automaton in Example~\ref{ex:bdd-automaton} corresponds, according to Example~\ref{ex:prop}, to the following symbolic automaton: 
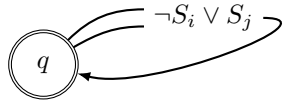
\begin{figure}[H]
\centering
\begin{tikzpicture}[
  arrows={[scale=.75]},
  every edge quotes/.append style={inner sep=+.15em, node font=\footnotesize}]
\node (q1) [state, accepting] {$q$};
\node (bb)[above right=0mm and 1cm of q1] {$\lnot S_i \lor S_j$};
\path[thick, in=180, at end]
  (q1) edge[out=45, "" above left] ([yshift= .75ex]bb.west)
       edge[out=30, "" below left] ([yshift=-.75ex]bb.west)
  (bb) edge[-Latex, out=0, in=-15] (q1);
\end{tikzpicture}
\end{figure}
\end{example}

To find the computational complexity of deciding formulae of the form~(\ref{eq:table-power}), consider first the case where the propositional formulae $L_1,\ldots,L_m$ for the automaton $M$ denote disjoint Venn regions. In such case, checking the satisfiability of formula~(\ref{eq:table-power}) reduces to determining whether there exists a symbolic tree table $s$ such that whenever the number of times a certain propositional letter occurs is non-zero, the corresponding Venn region interpreted according to (\ref{eq:1}) has a satisfiable defining formula. From this observation follows that our decision procedure will need to compute the so-called Parikh image of the regular tree language $M_S(L_1,\ldots,L_m)$. 

\begin{definition}[Parikh Image] \label{def:parikh} ~\\
The Parikh image of $M_S(L_1,\ldots,L_m)$ is the set 
\[
\parikh(M_S(L_1,\ldots,L_m)) = \{ (|s|_{L_1}, \ldots, |s|_{L_m}) | s \in M_S(L_1,\ldots,L_m) \} 
\]
where $|s|_{L_i}$ denotes the number of occurrences of the propositional formula $L_i$ in the symbolic tree table $s$. 
\end{definition}

To compute the Parikh image of a regular tree language, we will use a key observation by Klaedtke and Rue{\ss} \cite{klaedtke_parikh_2002} that allows to reduce the problem to that of computing the Parikh image of a context-free grammar.

First, note that the tree language $M_S(L_1,\ldots,L_m)$ is given by a non-deterministic bottom-up tree automaton (this follows from Definition~\ref{def:treeautomata}). However, the construction of Klaedtke and Rue{\ss} is given by non-deterministic top-down tree automata.

\begin{definition}
A non-deterministic top-down tree automaton is a tuple $\mathcal{A} = \left(Q, \Sigma, \delta, q_0, F\right)$, where 
\begin{itemize}
\item[-] $Q$ is a finite set of states.
\item[-] $\Sigma$ is an alphabet. 
\item[-] $\delta: Q \times \Sigma \rightarrow \mathcal{P}(Q \times Q)$ is the transition function. 
\item[-] $q_0$ is the initial state. 
\item[-] $F \subseteq Q$ is the set of final states. 
\end{itemize}

Associated with $\mathcal{A}$ is the function $\overline{\delta}: \Sigma^{\#} \to S$ defined by $\overline{\delta}(\Lambda) = \Lambda$ and 
\[
\overline{\delta}(\sigma[\tau,\tau']) = \delta(q_0,\sigma)
\]

A run $\varrho$ of $\mathcal{A}$ on $t \in T_{\Sigma}$ is a $Q$-labelled tree $\varrho$ with $\operatorname{dom}(\varrho)=\{\lambda\} \cup\{u b \mid u \in \operatorname{dom}(t), b \in\{0,1\}\}$ such that $\varrho(u)=q_0$, and $(\varrho(u 0), \varrho(u 1)) \in \delta(\varrho(u), t(u))$, for all $u \in \operatorname{dom}(t)$.  

$\varrho$ is accepting if all leaves of $\varrho$ are labeled with states in $F$, i.e. $\varrho(u) \in F$, for all $u \in \operatorname{dom}(\varrho) \backslash \operatorname{dom}(t)$. 

A tree $t$ is recognized by $\mathcal{A}$ if there is an accepting run of $\mathcal{A}$ on $t$. 

$T(\mathcal{A})$ denotes the set of trees that are recognized by $\mathcal{A}$.
\end{definition}

Fortunately, it is easy to convert from non-deterministic top-down to non-deterministic bottom-up tree automata.

\begin{proposition}[{\cite[Theorem~1.6.1]{comon_tree_2008}}]
The class of languages accepted by top-down NFTAs is precisely the class of languages accepted by bottom-up NFTAs. Given a top-down (bottom-up) NFTA one can compute a bottom-up (top-down) NFTA in linear time in the number of edges and states of the input. 
\end{proposition}

Second, we use the observation of Klaedtke and Rue{\ss} \cite[Lemma~17]{klaedtke_parikh_2002} to compute a context-free grammar with the same Parikh image.

\begin{lemma}[\cite{klaedtke_parikh_2002}] \label{lem:grammar}
For any non-deterministic top-down tree automaton $\mathcal{A}$ one can compute in linear time a context-free grammar $G_{\mathcal{A}}$ expressing the trees accepted by $\mathcal{A}$ as words obtained through the in-order traversal of the trees. As a consequence $\parikh(\mathcal{A}) = \parikh(G_{\mathcal{A}})$. 
\end{lemma}
\begin{proof}
Let $\mathcal{A}=\left(Q, \Gamma, \delta, q_{I}, F\right)$ be a top-down tree automaton. We define a context-free grammar $G = \langle V, \Sigma, R, S \rangle$ that generates the words obtained by traversing the trees recognized by $\mathcal{A}$ in infix order as follows:

\begin{itemize}
    \item $V = Q$ is the set of non-terminal symbols. 
    \item $\Sigma$ is the set of terminal symbols.
    \item There are two kinds of derivation rules: 

    \begin{itemize}
    \item For each $\left(q, q^{\prime}\right) \in \delta(p, b)$, we have the rule $p \rightarrow q b q^{\prime}$.
    \item If $(F \times F) \cap \delta(q, b) \neq \emptyset$ then we have the rule $q \rightarrow b$. 
    \end{itemize}

    \item $S = q_{I}$ is the start symbol of $G$.
\end{itemize}

It is immediate from the definition that $Inorder(L(\mathcal{A})) = L(G)$ and that the size of the grammar is equal to that of the automaton. Since the Parikh image is invariant under permutation of the labels, it follows that $\parikh(\mathcal{A}) = \parikh(G_{\mathcal{A}})$. 
\end{proof}

The second key observation, by Verma, Seidl and Schwentick, is that the Parikh image of a context-free grammar can be described by a linear-sized existential Presburger arithmetic formula.

\begin{lemma}[\cite{verma_complexity_2005}] 
\label{lem:verma}
Given a context-free grammar G, one can compute an existential Presburger formula $\phi_G$ for the Parikh image of $L(G)$ in linear time.
\end{lemma}

In summary, 

\begin{lemma}
\label{lem:parikhlinear} 
The set $\parikh(M_S(L_1,\ldots,L_n))$ is definable by an existential Presburger formula $\rho$ of size $O(|M|)$ where $|M|$ is the number of symbols used to describe the automaton $M$.
\end{lemma}

In the more general case, when propositional letters denote overlapping Venn regions, a partitioning argument is required. This is formalised in Theorem~\ref{thm:tableelimination}. First, we fix some notation. We set $p_{\beta} := \bigcap_{i = 1}^k S_i^{\beta_i}$ where $\beta \in \{0,1\}^k$, $p_L := \bigcup\limits_{\beta \models L} p_{\beta}$ where $L$ is a propositional formula and $\models$ is the propositional satisfaction relation that is true if and only if the assignment of the values in $\beta$ to the free variables in $L$ makes the formula $L$ true. When using the interpretation of sets of the form~(\ref{eq:1}), the formula defining the Venn region $p_{\beta}$ will be denoted by $\varphi^{\beta}(d) := \bigwedge_{i = 1}^k  \varphi^{\beta(i)}_i(d)$. We write $S_1 \dot{\cup} S_2$ to denote 
the set $S_1 \cup S_2$ where we want to emphasise that $S_1 \cap S_2 = \emptyset$. Finally, we use the notation $[n] := \set{1,\ldots,n}$ to refer to the first $n$ natural numbers.

We give next the technical statement of the main theorem of this section. The reader should refer to the explanations following the statement for the intuition behind it.

\begin{theorem} \label{thm:tableelimination}
Formula~(\ref{eq:table-power}) is equivalent to the formula
\begin{equation} \label{eq:elimprop}
\begin{split}
\exists s \in& [m]. \sigma: [s] \hookrightarrow [m].\exists \beta_{1}, \ldots,\beta_{s} \in \{0,1\}^k. \bigwedge_{j = 1}^s \exists d. \phi^{\beta_{j}}(d) \land \\ \exists k_1,&\ldots,k_m. \exists S_1,\ldots, S_k, P_1,\ldots,P_s. \\ & 
  \rho(k_1,\ldots, k_m) \land \bigwedge_{i = 1}^s P_i \subseteq p_{L_{\sigma(i)}} \land  \cup_{i = 1}^m p_{L_i} = \dot{\cup}_{i = 1}^s P_i \land \\ 
  &\bigwedge_{i = 1}^s |P_i| = k_{\sigma(i)} \land \bigwedge_{i = 1}^s p_{\beta_i} \cap P_{i} \neq \emptyset
\end{split}
\end{equation}
where $\sigma$ is an injection from $\set{1,\ldots,s}$ to $\set{1,\ldots,m}$, $\rho$ is the arithmetic expression in Lemma~\ref{lem:parikhlinear} and $|\cdot|$ denotes the cardinality of the argument set expression.
\end{theorem}

We start by observing that formula~(\ref{eq:elimprop}) has two parts. The first part corresponds to the subterm $\bigwedge_{j = 1}^s \exists d. \varphi^{\beta_{j}}(d)$ and falls within the existential theory of the elements in $\mathcal{D}$, $Th_{\exists^*}(\mathcal{D})$. The second part corresponds to the remaining subterm and falls within the quantifier-free first-order theory of Boolean Algebra with Presburger arithmetic ($\qfbapa$)\cite{kuncak_towards_2007}, which can be viewed as the monadic second order theory $Th_{\exists^*}^{mon}(\langle \mathbb{N}, \subseteq, \sim \rangle)$ where $\sim$ is the equicardinality relation between two sets. 

The second observation is that formula~(\ref{eq:elimprop}) is distilled from a non-deterministic decision procedure for the formulae of the shape (\ref{eq:table-power}). The existentially quantified variables $s, \sigma, \beta_{1}, \ldots,\beta_{s}$ are guessed by the procedure. These guessed values are then used by specialised procedures for $Th_{\exists^*}(\mathcal{D})$ and $Th_{\exists^*}^{mon}(\langle \mathbb{N}, \subseteq, \sim \rangle)$. For the convenience of the reader, we describe here what these values mean (this meaning follows from the proof of the theorem below). The value of $s$ represents the number of Venn regions associated to the formulae $L_1,\ldots,L_m$ that will be non-empty. $\sigma$ indexes these non-empty regions. $\beta_1,\ldots,\beta_s$ are elementary Venn regions contained in the non-empty ones. 

Observe that Theorem~\ref{thm:tableelimination} refines the statement in Lemma~\ref{lem:automatonlanguage}: the satisfiability problem of SFAs is decomposed into the decision problem of the existential fragment of the theory of the input characters and the existential fragment of the monadic second-order theory of the indices.

Finally, it remains to exemplify the situation in which the Venn regions overlap, which justifies the introduction of the partition variables $P_1,\ldots,P_s$ in formula~(\ref{eq:elimprop}).

\begin{example} 
\label{ex:partition}
Consider the situation where $S_1 \land S_2$ and $S_2 \land S_3$ are two propositional formulae labelling the transitions of the symbolic automaton. These formulae correspond to the Venn regions $S_1 \cap S_2$ and $S_2 \cap S_3$, which share the region $S_1 \cap S_2 \cap S_3$. Given a model of $S_1, S_2$ and $S_3$, how do we guarantee that the indices in the region $S_1 \cap S_2 \cap S_3$ are consistent with a run of the automaton? For instance, the automaton may require one element in $S_1 \cap S_2$ and another in $S_2 \cap S_3$. Placing a single index in $S_1 \cap S_2 \cap S_3$ would satisfy the overall cardinality constraints, but not the fact that overall we need to have two elements. Trying to specify these restrictions in the general case would reduce to specifying an exponential number of cardinalities. 
\end{example}

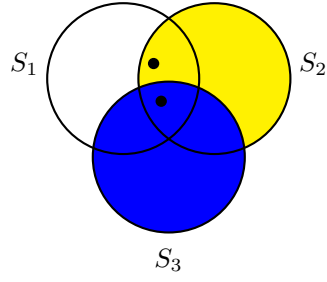
\begin{figure} 
\centering
\begin{tikzpicture}[thick]
\fill[yellow] (1.2,0) circle (1);
\fill[blue] (.6,-1.04) circle (1);

\draw (0,0) circle (1) node[above,shift={(0,1)}] {};
\draw (1.2,0) circle (1) node[above,shift={(0,1)}] {};
\draw (.6,-1.04) circle (1) node[shift={(1.1,-.6)}] {};

\node at (-1.3,.2) {$S_1$};
\node at (2.5,.2) {$S_2$};
\node at (0.6,-2.4) {$S_3$};
\node at (.4,.2) [circle,fill,inner sep=1.5pt]{};
\node at (.5,-.3) [circle,fill,inner sep=1.5pt]{};
\end{tikzpicture}
\caption{A Venn diagram representing the situation discussed in Example~\ref{ex:partition}. To deal with the overlapping regions, it is necessary to decide to which sets do the different indices belong to. In the Figure, the different parts are marked with colours.}
\label{fig:partition}
\end{figure}

We proceed next to the proof of the theorem. 

\begin{proof}[Proof of Theorem~\ref{thm:tableelimination}]
$\Rightarrow)$ If formula~(\ref{eq:table-power}) is satisfiable, then there are sets $S_1,\ldots,S_k$, a word $d$ and a tree table $\tau$ satisfying
\begin{equation*} \label{eq:te1} 
\bigwedge_{i = 1}^k S_i = \Set{ n \in \{0,1\}^* | \phi_i(d) } \land \tau \in M(L_1,\ldots,L_s) \land \bigwedge_{i = 1}^k S_i = \Set{ n \in \{0,1\}^* | t_i(n) }
\end{equation*}

Let $\overline{s} \in M(L_1,\ldots,L_s)$ be the symbolic tree table corresponding to $\tau$ (where the bar is used to distinguish it from variable $s$ appearing in the formula~(\ref{eq:elimprop})). We define $k_i := |\overline{s}|_{L_i}, s = |\Set{ i | k_i \neq 0 }|$, $\sigma$ mapping the indices in $[s]$ to the indices of the terms for which $k_i$ is non-zero and $P_i = \set{ n \in \{0,1\}^* | \overline{s}(n) = L_{\sigma(i)} }$. It will be convenient to work out the following equalities:
\begin{align} \label{eq:te2}
\begin{split}
p_{L_i} &= \bigcup_{\beta \models L_i} \bigcap_{j = 1}^k S_j^{\beta_j} = \bigcup_{\beta \models L_i} \Set{ n \in \{0,1\}^* | \bigwedge_{j = 1}^k t_j^{\beta_j}(n) } \\ &  =  \Set{ n \in \{0,1\}^* | \tau(n) \models L_i } \\
p_{L_i} &= \bigcup_{\beta \models L_i} \bigcap_{j = 1}^k S_j^{\beta_j} = \bigcup_{\beta \models L_i} \Set{ n \in \mathcal{D} | \bigwedge_{j = 1}^k \phi_j^{\beta_j}(d) } \\ & =  \Set{ d \in \mathcal{D} | L_i(\overline{\phi}(d)) }
\end{split}
\end{align}
where $L_i(\overline{\phi}(d(n)))$ is the propositional formula obtained by substituting each set variable $S_i$ by the formulae $\phi_i(d(n))$. We now deduce formula~(\ref{eq:elimprop}):
\begin{itemize}
    \renewcommand\labelitemi{-}
    \item $\rho(k_1,\ldots,k_m)$: from $\overline{s} \in P(L_1,\ldots,L_m)$, we have that  
    \[
    (k_1,\ldots,k_m) \in \parikh(M_S(L_1,\ldots,L_m))
    \]
    and therefore $\rho(k_1,\ldots,k_m)$.
    \item $P_i \subseteq p_{L_{\sigma(i)}}$: since $\overline{s}$ corresponds to $\tau$, for all $n \in \N$ we have $\tau(n) \models \overline{s}(n)$ and the inclusion follows from the definition of $P_i$ and equation~(\ref{eq:te2}).
    \item $|P_i| = k_{\sigma(i)}$: since 
    \[
    |P_i| = \Big|\Set{ n \in \{0,1\}^* | \overline{s}(n) = L_{\sigma(i)} }\Big| = |\overline{s}|_{L_{\sigma(i)}} = k_{\sigma(i)}
    \]
    \item Each pair of sets $P_i, P_j$ with $i < j$ is disjoint: 
    \begin{align*}
    P_i \cap P_j &= \Set{ n \in \{0,1\}^* | \overline{s}(n) = L_{\sigma(i)} } \cap \Set{ n \in \{0,1\}^* | \overline{s}(n) = L_{\sigma(j)} } = \\ &= \Set{ n \in \{0,1\}^* | \overline{s}(n) = L_{\sigma(i)} = L_{\sigma(j)} } = \emptyset
    \end{align*}
    using that the letters $L$ are chosen to be distinct and that $\sigma$ is an injection (so $\sigma(i) \neq \sigma(j)$). 

    \item $p_{L_1} \cup \ldots \cup p_{L_m} = P_1 \dot{\cup} \ldots \dot{\cup} P_s$: since by definition 
    \[
    P_i = \Set{ n \in \{0,1\}^* | \overline{s}(n) = L_{\sigma(i)} }, p_{L_i} = \Set{ n \in \{0,1\}^* | \tau(n) \models L_i }
    \]
    and by definition of $\sigma$ it follows that the only letters that can appear in $\overline{s}$ are $L_{\sigma(1)}, \ldots, L_{\sigma(s)}$. Thus, we have $p_{L_1} \cup \ldots \cup p_{L_m} = [1,|\overline{t}|] = [1,|\overline{s}|] = P_1 \dot{\cup} \ldots \dot{\cup} P_s$. 
    
    \item There exists $\beta_1,\ldots,\beta_s \in \{0,1\}^k$, such that $\bigwedge_{i = 1}^s P_{\beta_i} \cap P_{i} \neq \emptyset$: note that $P_i \neq \emptyset$ by definition of $\sigma$. Thus, there must exist some $\beta_i$ such that $p_{\beta_i} \cap P_{i} \neq \emptyset$. We pick any such $\beta_i$.
    \item $\bigwedge_{j = 1}^s \exists d.  \varphi^{\beta_j}(d)$: follows from $p_{\beta_{j}} \cap P_j \neq \emptyset$ and formula~(\ref{eq:te2}).
\end{itemize}

$\Leftarrow)$ Conversely, if formula~(\ref{eq:elimprop}) is satisfiable, then there is an integer $s \in [n]$, an injection $\sigma: [s] \hookrightarrow [n]$, bit-strings $\beta_{1},\ldots,\beta_{s} \in \{0,1\}^k$, integers $k_1,\ldots,k_m$ and sets $S_1,\ldots,S_k,P_1,\ldots,P_s$ satisfying  
\begin{equation} 
\begin{split}
&\bigwedge_{j = 1}^s \exists d.   \varphi^{\beta_{j}}(d) \land 
  \rho(k_1,\ldots, k_m) \land \bigwedge_{i = 1}^s P_i \subseteq p_{L_{\sigma(i)}} \land  \cup_{i = 1}^m p_{L_i} = \dot{\cup}_{i = 1}^s P_i \land \\
 &\bigwedge_{i = 1}^s |P_i| = k_{\sigma(i)} \land \bigwedge_{i = 1}^s p_{\beta_{i}} \cap P_{i} \neq \emptyset 
\end{split}
\end{equation}
From $
\rho(k_1,\ldots, k_n)
$ follows that there is a symbolic table $\overline{s} \in M_S(L_1,\ldots,L_m)$ such that $|\overline{s}|_{L_i} = k_i$ for each $L_i \in \set{L_1,\ldots,L_m}$. From formula~(\ref{eq:te2}) and
\begin{equation*}
p_{L_1} \cup \ldots \cup p_{L_m} = P_1 \dot{\cup} \ldots \dot{\cup} P_s \land 
\bigwedge_{i = 1}^s P_{i} \subseteq p_{L_{\sigma(i)}} \land \bigwedge_{i = 1}^s |P_i| = k_{\sigma(i)} 
\end{equation*}
follows that we can replace the formulae $L_i$ occurring in the symbolic table $\overline{s}$ by the bit-strings representing the elementary Venn regions to which the indices of the sets $P_i$ belong. Moreover, thanks to the condition $\bigwedge_{i = 1}^s p_{\beta_{i}} \cap P_{i} \neq \emptyset$ follows that we can replace the letters $L_i$ by the bit-strings $\beta_i$, defining $\tau$ as $\tau(n) = 
\begin{cases}
\beta_{i} & \text{if } n \in P_{i}
\end{cases}$. In this way, we obtain a table $\tau \in M(L_1,\ldots,L_s)$. We then define the corresponding word over $\mathcal{D}$, thanks to the property $\bigwedge_{i = 1}^s \exists d. \phi^{\beta_{i}}(d)$. Naming the witnesses of these formulae as $d_{i}$, we define $d(n) = 
\begin{cases}
d_{i} & \text{if } n \in P_{i}
\end{cases}$. To conclude, note that:
\[
\Set{ n \in \{0,1\}^* | t_j(n) } = \cup_{\Set{ 1 \le i \le k | \beta_{i}(j) = 1}} P_{i} = \Set{ n \in \{0,1\}^* | \phi_j(d(n)) }
\]
Thus, we have that formula~(\ref{eq:table-power}) is satisfied by the set variables 
\[
S_j := \Set{ n \in \{0,1\}^* | t_j(n) } = \Set{ n \in \{0,1\}^* | \phi_j(d(n)) }
\]
\end{proof}

\section{Quantifier-free Boolean Algebra with Presburger Arithmetic}
\label{section:with}
The arguments following the statement of Theorem~\ref{thm:tableelimination} sketch a non-deterministic procedure for the satisfiability problem of symbolic finite automata, based on the existence of decision procedures for $Th_{\exists^*}(\mathcal{D})$ and $Th_{\exists^*}^{mon}(\langle \mathbb{N}, \subseteq, \sim \rangle)$. In this section, we recall the non-deterministic polynomial time decision procedure for $Th_{\exists^*}^{mon}(\langle \mathbb{N}, \subseteq, \sim \rangle)$. As a consequence, we obtain Corollary~\ref{cor:complexity} which situates the decision problem of symbolic finite automata in the classical complexity hierarchy. This section should also prepare the reader for the extension of these results, where the automaton can require linear arithmetic constraints on the cardinalities of the effective Boolean algebra. This extension is described in Section~\ref{section:cardinalities}. 

Instead of working with $Th_{\exists^*}^{mon}(\langle \mathbb{N}, \subseteq, \sim \rangle)$ directly, we use the logic $\qfbapa$ \cite{kuncak_towards_2007} which has the same expressive power \cite[Section~2]{kuncak_deciding_2006}. The syntax of $\qfbapa$ is given in Figure~\ref{fig:qfbapa-syntax}. The meaning of the syntax is as follows. $F$ presents the Boolean structure of the formula, $A$ stands for the top-level constraints, $B$ gives the Boolean restrictions and $T$ the Presburger arithmetic terms. The operator $\text{dvd}$ stands for the divisibility relation and $\mathcal{U}$ represents the universal set. The remaining interpretations are standard.

\begin{figure}[!ht]
\centering
\begin{align*}
F & ::= A \, | \, F_1 \land F_2 \, | \, F_1 \lor F_2 \, | \, \lnot F \\
A & ::= B_1 = B_2 \, | \, B_1 \subseteq B_2 \, | \, T_1 = T_2 \, | \, T_1 \le T_2 \, | \, K \text{ dvd } T \\
B & ::= x \, | \, \emptyset \, | \, \mathcal{U} \, | \, B_1 \cup B_2 \, | \, B_1 \cap B_2 \, | \, B^c \\
T & ::= k \, | \, K \, | \,  T_1 + T_2 \, | \, K \cdot T \, |  \, |B| \\
K & ::= \ldots \, | \, -2 \, | \, -1 \, | \, 0 \, | \, 1 \, | \, 2 \, | \, \ldots
\end{align*}
\caption{$\qfbapa$'s syntax}
\label{fig:qfbapa-syntax}
\end{figure}

The satisfiability problem of this logic is reducible to propositional satisfiability in polynomial time. Our proofs will rely on the method of \cite{kuncak_towards_2007}, which we sketch briefly here. The basic argument to establish a $\np$ complexity bound on the satisfiability problem of $\qfbapa$ is based on a theorem by Eisenbrand and Shmonin \cite{eisenbrand_caratheodory_2006}, which in our context says that any element of an integer cone can be expressed in terms of a polynomial number of generators. Figure~\ref{fig:pa-verifier} gives a verifier for this basic version of the algorithm. The algorithm uses an auxiliary verifier $V_{PA}$ for the quantifier-free fragment of Presburger arithmetic. The key step is showing equisatisfiability between 2.(b) and 2.(c). If $x_1, \ldots, x_k$ are the variables occurring in $b_0, \ldots, b_p$ then we write $p_\beta = \bigcap\limits_{i = 1}^k x_i^{e_i}$ for $\beta = (e_1,\ldots,e_k) \in \{0,1\}^k$ where we define $x^1 := x$ and $x^0 := \mathcal{U} \setminus x$ as before. If we define $\llbracket b_i \rrbracket_{\beta_j}$ as the evaluation of $b_i$ as a propositional formula with the assignment given in $\beta$ and introduce variables $l_\beta = |p_\beta|$, then $|b_i| =  \sum\limits_{j = 0}^{2^e-1} \llbracket b_i \rrbracket_{\beta_j} l_{\beta_j}$, so the restriction $\bigwedge\limits_{i = 0}^p |b_i| = k_i$ in 2.(b) becomes $\bigwedge\limits_{i = 0}^p \sum\limits_{j = 0}^{2^e-1} \llbracket b_i \rrbracket_{\beta_j} l_{\beta_j} = k_i$ which can be seen as a linear combination in the set of vectors $
\{(\llbracket b_0 \rrbracket_{\beta_j}, \ldots, \llbracket b_p \rrbracket_{\beta_j}). j \in \{0, \ldots, 2^e-1\}\} 
$, i.e. as
\[
\bigwedge\limits_{i = 0}^p \sum\limits_{j = 0}^{2^e-1} 
\begin{pmatrix}
\llbracket b_0 \rrbracket_{\beta_j} \\ \vdots \\ \llbracket b_p \rrbracket_{\beta_j}
\end{pmatrix}
l_{\beta_j} = k_i
\]
Eisenbrand-Shmonin's result allows then to derive 2.(c) for $N$ polynomial in $|x|$. In the other direction, it is sufficient to set $l_{\beta_j} = 0$ for $j \in \{0, \ldots, 2^e-1\} \setminus \{i_1, \ldots, i_N\}$. Thus, we have:

\begin{theorem}[{\cite{kuncak_towards_2007}}] \label{thm:qfbapa-complexity}
The satisfiability problem of $\qfbapa$ is in NP. 
\end{theorem}

From Theorems~\ref{thm:tableelimination} and \ref{thm:qfbapa-complexity}, we obtain the following improvement of \cite[Theorem~2]{veanes_symbolic_2015}:

\begin{corollary} \label{cor:complexity}
Let $Th_{\exists^*}(\mathcal{D})$ be the existential first-order theory of the formulae used in the transitions of the symbolic tree automaton $M$. If $Th_{\exists^*}(\mathcal{D}) \in \cc$ for some $\cc \supseteq \np$ then $\mathcal{L}(M) \neq \emptyset \in \cc$. 
\end{corollary}
\begin{proof}
The procedure non-deterministically guesses the value of the variables $s,\sigma,\beta_1,\ldots,\beta_s$ and uses a decision procedure for $Th_{\exists^*}(\mathcal{D})$ and a non-deterministic polynomial time decision procedure for $\qfbapa$ to check the corresponding sub-formulae in (\ref{eq:elimprop}). The correctness of the procedure follows from Theorem~\ref{thm:tableelimination}.
\end{proof}

Observe that in typical examples, $Th_{\exists^*}(\mathcal{D}) \in \np$ and thus, from Corollary~\ref{cor:complexity} it follows that $\mathcal{L}(M) \neq \emptyset \in \np$. This partially explains the success in the automation of SFAs in SMT solvers, which rely on solvers for propositional satisfiability.

\section{Decision Procedure for Non-Emptiness with Cardinalities}
\label{section:cardinalities}
We now consider a generalisation of the language of a symbolic tree automaton from Lemma~\ref{lem:automatonlanguage} with cardinality constraints on the effective Boolean algebra. Similar extensions for
related models of automata are considered in the literature on data words \cite{figueira_reasoning_2022}.

\begin{definition}
A symbolic tree automaton  with cardinalities accepts a language  of the form:
\[
\mathcal{L}(M) = \left 
\{ d \in \mathcal{D}^* \middle \vert \begin{array}{l}
F(S_1,\ldots,S_k) \land \bigwedge_{i = 1}^k S_i = \Set{ n \in \{0,1\}^* | \phi_i(d(n)) } \land \\ \exists \tau \in M(L_1,\ldots,L_m). \bigwedge_{i = 1}^k S_i = \Set{ n \in \{0,1\}^* | \tau_i(n) } 
\end{array}
\right\}
\]
where $F$ is a formula from $\qfbapa$. 
\end{definition}

Thus, checking non-emptiness of the language of a symbolic tree automaton with cardinalities reduces to checking whether the following formula is true:
\begin{equation} \label{eq:table-power-card}
\begin{split}
\exists S_1, \ldots, S_k. &F(S_1,\ldots,S_k) \land \\ &\exists d. \bigwedge_{i = 1}^k S_i = \Set{ n \in \{0,1\}^* | \phi_i(d(n)) } \land \\ &\exists \tau \in M(L_1,\ldots,L_m) \land \bigwedge_{i = 1}^k S_i = \Set{ n \in \{0,1\}^* | \tau_i(n) }
\end{split}
\end{equation}

To show that Theorem~\ref{thm:tableelimination} and Corollary~\ref{cor:complexity} stay true with linear arithmetic constraints on the cardinalities, we need to repeat part of the argument in Theorem~\ref{thm:tableelimination} since if $F$ denotes the newly introduced $\qfbapa$ formula and $G,H$ are the formulae shown equivalent in Theorem~\ref{thm:tableelimination}, then, from: 
\[
\exists S_1,\ldots,S_k. F(S_1,\ldots,S_k) \land G(S_1,\ldots,S_k)
\]
and
\[
\Big[ \exists S_1,\ldots,S_k. G(S_1,\ldots,S_k) \Big] \iff \Big[ \exists S_1,\ldots,S_k. H(S_1,\ldots,S_k) \Big]
\]
it does not follow that: 
\[
\exists S_1,\ldots,S_k. F(S_1,\ldots,S_k) \land H(S_1,\ldots,S_k) 
\]
Instead, the algorithm derives the cardinality constraints from each theory and then uses the sparsity of solutions \textit{over the satisfiable regions}. In the proof, we use the notations $\llbracket b_i \rrbracket_{\beta_j}$ and $l_{\beta}$ introduced in Section~\ref{section:with}.

\begin{theorem}
\label{thm:card}
Formula~(\ref{eq:table-power-card}) is equivalent to:
\begin{equation} \label{eq:elimpropcard}
\begin{split}
\exists N \le p(|F|), &\exists s \in [m]. \sigma: [s] \hookrightarrow [m]. 
\exists \beta_{1}, \ldots,\beta_{N} \in \{0,1\}^k. \bigwedge_{j = 1}^N \exists d. \phi^{\beta_{j}}(d) \land \\ \exists k_1,\ldots,k_m.& \exists S_1,\ldots, S_k, P_1,\ldots,P_s. \\ 
  & \rho(k_1,\ldots, k_m) \land \bigwedge_{i = 1}^s P_i \subseteq p_{L_{\sigma(i)}} \land  \cup_{i = 1}^m p_{L_i} = \dot{\cup}_{i = 1}^s P_i \land \\ 
  &\bigwedge_{i = 1}^s |P_i| = k_{\sigma(i)} \land \cup_{i = 1}^N p_{\beta_i} = \dot{\cup}_{i = 1}^s P_{i} \land F(S_1,\ldots,S_k)
\end{split}
\end{equation}
where $p$ is a polynomial and $|F|$ is the number of symbols used to write $F$. 
\end{theorem}
\begin{proof}
The proof is deferred to the appendix.
\end{proof}

We can thus formulate the analogous to Corollary~\ref{cor:complexity} in the case of finite symbolic automata with cardinalities.

\begin{corollary} \label{cor:complexitycard}
Let $Th_{\exists^*}(\mathcal{D})$ be the existential first-order  theory of the formulae used in the transitions of a symbolic finite automaton with cardinality constraints. If $Th_{\exists^*}(\mathcal{D}) \in \cc$ for some $\cc \supseteq \np$ then $\mathcal{L}(M) \neq \emptyset \in \cc$. 
\end{corollary}
\begin{proof}
As in Corollary~\ref{cor:complexity}.
\end{proof}

\section{Conclusion}
\label{section:conclusion}
We have revisited the model of symbolic tree automata as it was introduced in \cite{veanes_symbolic_2015}. We have obtained tight complexity bounds on their non-emptiness problem. Our methodology follows the Feferman-Vaught decomposition technique in that it reduces the non-emptiness problem of the automaton to the satisfiability problem of the existential first-order theory of the characters accepted by the automaton and the satisfiability problem of the existential monadic second-order theory of the indices. 

To combine these two distinct theories we use the ideas from the combination method through sets and cardinalities of Wies, Piskac and Kun\v{c}ak \cite{wies_combining_2009} and the computation of an equivalent linear-sized existentially quantified Presburger arithmetic formula from the Parikh image of a regular tree language. The latter combines two observations. The first observation by Klaedtke and Rue{\ss} \cite{klaedtke_parikh_2002} connects this problem with the computation of the Parikh image of a context-free grammar. The second observation by Verma, Seidl and Schwentick \cite{verma_complexity_2005} allows computing the Parikh image of a context-free grammar in terms of a linear-sized Presburger arithmetic formula. A crucial step in the proofs is a partitioning argument for the underlying Venn regions. We profit from the analysis in \cite{kuncak_towards_2007} to extend our arguments to the satisfiability problem of finite symbolic automata that consider linear arithmetic restrictions over the cardinalities of the Boolean algebra associated with the symbolic finite automaton. 

In future work, we plan to extend our methods to other variants of symbolic automata to which we believe similar techniques may be applicable. Another interesting research direction would be to consider extensions of the language that allow free variables in set interpretations of the form~(\ref{eq:1}), which seems to have applications to various satisfiability problems. Recently, Hague et alii. have found some \textit{parallel} with the results in this paper \cite{hague_parikhs_2023}. However, our main motivation was the application of the Feferman-Vaught decidability technique in the bounded complexity setting. In particular, this work shows the usefulness of taking the \textit{reduction sequence} of the Feferman-Vaught theorem to be a \textit{partitioning sequence} \cite[Theorem~3.1]{feferman_first_1959}. This was not completely evident at first, and improves over the results in 
\cite{raya_algebraic_2024}, by allowing the ordering relation to range over non-disjoint regions. A natural continuation of our work would be to find similar decompositions for different decision problems of interest for symbolic automata.

\newpage
\appendix
\section{Verifier for $\qfbapa$}
\begin{figure*}[ht!]
\fbox{\parbox{.95\textwidth}{
On input $\langle x, w \rangle$:

\begin{enumerate}
\setlength\itemsep{1em}
\item Interpret $w$ as:

\vspace{1em}

\begin{enumerate}
    \item a list of indices $i_1, \ldots, i_N \in \{0, \ldots, 2^e-1 \}$ where $e$ is the number of set variables in $x$.
    \item a certificate $C$ for $V_{PA}$ on input $x'$ defined below.
\end{enumerate}

\vspace{.5em}

\item Transform $x$ into $x'$ by:

\vspace{1em}

\begin{enumerate}
    \item rewriting boolean expressions according to the rules: \begin{align*}
    b_1 = b_2 & \mapsto b_1 \subseteq b_2 \land b_2 \subseteq b_1 \\
    b_1 \subseteq b_2 & \mapsto |b_1 \cap b_2^c| = 0
    \end{align*}
    
    \item introducing variables $k_i$ for cardinality expressions: $$G \land \bigwedge_{i = 0}^{p} |b_i| = k_i$$ where $G$ is the resulting quantifier-free Presburger arithmetic formula.
    
    \item rewriting into:
    $$ G \land \bigwedge\limits_{j = i_1, \ldots, i_N} l_{\beta_j} \ge 0 \land \bigwedge_{i = 0}^{p} \sum_{j = i_1, \ldots, i_N} \llbracket b_i \rrbracket_{\beta_j} \cdot l_{\beta_j} = k_i$$
\end{enumerate}

\item Run $V_{PA}$ on $\langle x', C \rangle$.

\item Accept iff $V_{PA}$ accepts.
\end{enumerate}
}}
\caption{Verifier for $\qfbapa$}
\label{fig:pa-verifier}
\end{figure*}

\section{Proof of the Extension with Cardinalities}
\begin{proof}[Proof of Theorem~\ref{thm:card}]
$\Rightarrow)$ If formula~(\ref{eq:table-power-card}) is true, then there are sets $S_1,\ldots,S_k$, a finite tree $d$ and a tree table $\tau$ such that:
\begin{align} 
\label{eq:te2card}
\begin{split}
&F(S_1,\ldots,S_k) \land \bigwedge_{i = 1}^k S_i = \Set{ n \in \{0,1\}^* | \phi_i(d(n)) } \land \\ 
&\tau \in M(L_1,\ldots,L_m) \land \bigwedge_{i = 1}^k S_i = \Set{ n \in \{0,1\}^* | \tau_i(n) }
\end{split}
\end{align}
Thus, there exists a symbolic table $\overline{s} \in M_S(L_1,\ldots,L_s)$ corresponding to $\tau$. We define $k_i := |\overline{s}|_{L_i}, s = |\Set{ i | k_i \neq 0 }|$, $\sigma$ maps the indices in $[s]$ to the indices of the terms for which $k_i$ is non-zero and $P_i = \Set{ n \in \{0,1\}^* | \overline{s}(n) = L_{\sigma(i)} }$. As in Theorem~\ref{thm:tableelimination}, we have the equalities $p_{L_i} = \Set{ n \in \{0,1\}^* | \tau(n) \models L_i }$, $p_{L_i} = \Set{ n \in \{0,1\}^* | L_i(\overline{\phi}(d)) }$ and we can show that the following formula holds:
\begin{equation} \label{eq:derivedcond}
\begin{split}
&\rho(k_1,\ldots, k_m) \land \bigwedge_{i = 1}^m P_i \subseteq p_{L_{\sigma(i)}} \land \cup_{i = 1}^m p_{L_i} = \dot{\cup}_{i = 1}^s P_i \land \\ & \bigwedge_{i = 1}^s |P_i| = k_{\sigma(i)} \land F(S_1,\ldots,S_k)
\end{split}
\end{equation}
We need to find a sparse model of (\ref{eq:derivedcond}). To achieve this, we follow the methodology in Theorem~\ref{thm:qfbapa-complexity}. This leads to a system of equations of the form:
\[
\exists c_1,\ldots,c_p. G \land \sum_{j = 0}^{2^e - 1} 
\begin{pmatrix}
\llbracket b_0 \rrbracket_{\beta_j} \\
\cdots \\
\llbracket b_p \rrbracket_{\beta_j}
\end{pmatrix}\cdot l_{\beta_j} = 
\begin{pmatrix}
c_1 \\
\ldots \\
c_p
\end{pmatrix}
\]
We remove those elementary Venn regions where $l_{\beta} = 0$. This includes regions whose associated formula in the interpreted Boolean algebra is unsatisfiable, and regions corresponding to tree table entries not occurring in $\tau$. This transformation gives a reduced set of indices $\mathcal{R}$ participating in the sum. 

Using Eisenbrand-Shmonin's theorem, we have a polynomial (in the size of the original formula) family of Venn regions $\beta_1, \ldots,\beta_N$ and corresponding cardinalities $l_{\beta_1}',\ldots,l_{\beta_N}'$, which we can assume to be non-zero, such that
\begin{equation} \label{eq:reduced}
\exists c_1,\ldots,c_p. G \land \sum_{\beta \in \set{\beta_1, \ldots,\beta_N} \subseteq \mathcal{R}}
\begin{pmatrix}
\llbracket b_0 \rrbracket_{\beta_j} \\
\cdots \\
\llbracket b_p \rrbracket_{\beta_j}
\end{pmatrix}\cdot l_{\beta_j}' = 
\begin{pmatrix}
c_1 \\
\ldots \\
c_p
\end{pmatrix}
\end{equation}


The satisfiability of formula~(\ref{eq:reduced}) implies the existence of sets of indices $p_{\beta}'$ satisfying the conditions derived in formula~(\ref{eq:derivedcond}). However, it does not imply which explicit indices belong to these sets and which are the contents corresponding to each index. From the condition 
\begin{equation*}
\rho(k_1,\ldots, k_n) \land \bigwedge_{i = 1}^n P_i' \subseteq p_{L_{\sigma(i)}}' \land \cup_{i = 1}^n p_{L_i}' = \dot{\cup}_{i = 1}^s P_i' \land \bigwedge_{i = 1}^s |P_i'| = k_{\sigma(i)} 
\end{equation*}
follows that there is a symbolic tree table $\overline{s}'$ satisfying $M_S(L_1,\ldots,L_n)$ with $k_{\sigma(i)}$ letters $L_{\sigma(i)}$ and that these letters are made concrete by entries in $P_i'$ for each $i \in \{1,\ldots,s\}$. We take the Venn regions $\beta \in \{\beta_1,\ldots,\beta_N\}$ such that $P_i' \supseteq p_{\beta}$ and label the corresponding entries in $\overline{s}'$ with $\beta$. In this way, we obtain a corresponding concrete tree table $\tau'$. This makes the indices in each Venn region concrete. To make the contents of the indices concrete, note that for each $\beta \in \mathcal{R}$, since $l_\beta \neq 0$, the formula $\exists d.
\phi^{\beta}(d)$ is true. In particular, this applies to each $\beta \in \{\beta_1, \ldots, \beta_N\}$. Thus, we obtain witnesses $d_1,\ldots,d_N$. We form a tree by replacing each letter $\beta$ in $\tau'$ by the corresponding value $d_{\beta}$. Then we have that formula~(\ref{eq:elimpropcard}) holds too. 

\vspace{.5em}

$\Leftarrow)$ If formula~(\ref{eq:elimpropcard}) is true, then there is $N \le p(|F|)$ where $p$ is a polynomial, $s \in [m]$, $\beta_1,\ldots,\beta_N \in \{0,1\}^k, k_1,\ldots,k_m \in \N$ and sets $S_1,\ldots,S_k, P_1,\ldots, P_s$ such that
\begin{align*}
\bigwedge_{j = 1}^N \exists d. \phi^{\beta_{j}}(d) \land &\rho(k_1,\ldots, k_m) \land \bigwedge_{i = 1}^s P_i \subseteq p_{L_{\sigma(i)}} \land  \cup_{i = 1}^m p_{L_i} = \dot{\cup}_{i = 1}^s P_i \land \\ 
  &\bigwedge_{i = 1}^s |P_i| = k_{\sigma(i)} \land \cup_{i = 1}^N p_{\beta_i} = \dot{\cup}_{i = 1}^s P_{i} \land F(S_1,\ldots,S_k)
\end{align*}

From $
\rho(k_1,\ldots, k_n)
$ follows that there is a symbolic table $\overline{s} \in M(L_1,\ldots,L_m)$ such that $|\overline{s}|_{L_i} = k_i$ for each $L_i \in \set{L_1,\ldots,L_m}$. From formula~(\ref{eq:te2card}) and
\begin{equation*}
p_{L_1} \cup \ldots \cup p_{L_m} = P_1 \dot{\cup} \ldots \dot{\cup} P_s \land 
\bigwedge_{i = 1}^s P_{i} \subseteq p_{L_{\sigma(i)}} \land \bigwedge_{i = 1}^s |P_i| = k_{\sigma(i)} 
\end{equation*}
follows that we can replace the formulae $L_i$ occurring in the symbolic table $\overline{s}$ by the bit-strings representing the elementary Venn regions to which the indices of the sets $P_i$ belong. Moreover, thanks to the condition $\cup_{i = 1}^N p_{\beta_i} = \dot{\cup}_{i = 1}^s P_{i}$, it follows that we can replace the letters $L_i$ by the bit-strings $\beta_i$. In this way, we obtain a table $\tau \in M(L_1,\ldots,L_m)$. We then define the corresponding word over $\mathcal{D}$, thanks to the property $
\bigwedge_{i = 1}^N \exists d. \phi^{\beta_{i}}(d) 
$. To conclude, note that:
\[
\Set{ n \in \{0,1\}^* | \tau_j(n) } = \cup_{\Set{ i | \beta_{i}(j) = 1}} P_{i} = \Set{ n \in \{0,1\}^* | \phi_j(d(n)) }
\]
Thus, we have that formula~(\ref{eq:table-power-card}) is satisfied by the set variables 
\[
S_j := \Set{ n \in \{0,1\}^* | \tau_j(n) } = \Set{ n \in \{0,1\}^* | \phi_j(d(n)) }
\]
\end{proof}

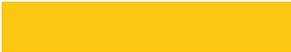
\end{document}